\documentclass[onecolumn,showpacs]{revtex4}
\usepackage{color}
\usepackage{graphicx}
\usepackage{dcolumn}
\usepackage{amsmath}
\usepackage{amssymb}
\usepackage{hyperref}
\usepackage{scalerel}
\newtheorem{theorem}{Theorem}[section]
\newtheorem{corollary}[theorem]{Corollary}
\newtheorem{proposition}[theorem]{Proposition}

\newtheorem{definition}[theorem]{Definition}
\newenvironment{proof}[1][Proof]{\begin{trivlist}
\item[\hskip \labelsep {\bfseries #1}]}{\end{trivlist}}

\newcommand{\qed}{\nobreak \ifvmode \relax \else
	\ifdim\lastskip<1.5em \hskip- \lastskip
	\hskip1.5em plus0em minus0.5em \fi \nobreak
	\vrule height0.75em width0.5em depth0.25em\fi}

\begin{document}

\title{Geometric properties of a certain class of compact dynamical horizons in locally rotationally symmetric class II spacetimes}

\author{Abbas \surname{Sherif}}
\email{abbasmsherif25@gmail.com}
\affiliation{Cosmology and Gravity Group, Department of Mathematics and Applied Mathematics, University of Cape Town, Rondebosch 7701, South Africa}

\author{Peter K. S. \surname{Dunsby}}
\email{peter.dunsby@uct.ac.za}
\affiliation{Cosmology and Gravity Group, Department of Mathematics and Applied Mathematics, University of Cape Town, Rondebosch 7701, South Africa\\South African Astronomical Observatory, Observatory 7925, Cape Town, South Africa}

\begin{abstract}
In this paper we study the geometry of a certain class of compact dynamical horizons with a time-dependent induced metric in locally rotationally symmetric class II spacetimes. We first obtain a compactness condition for embedded \(3\)-manifolds in these spacetimes, satisfying the weak energy condition, with non-negative isotropic pressure \(p\). General conditions for a \(3\)-manifold to be a dynamical horizon are imposed, as well as certain \textit{genericity} conditions, which in the case of locally rotationally symmetric class II spacetimes reduces to the statement that `the weak energy condition is strictly satisfied or otherwise violated'. The compactness condition is presented as a spatial first order partial differential equation in the sheet expansion \(\phi\), in the form \(\hat{\phi}+(3/4)\phi^2-cK=0\), where \(K\) is the Gaussian curvature of \(2\)-surfaces in the spacetime and \(c\) is a real number parametrizing the differential equation, where \(c\) can take on only two values, \(0\) and \(2\). Using geometric arguments, it is shown that the case \(c=2\) can be ruled out, and the \(\mathbb{S}^3\) (\(3\)-dimensional sphere) geometry of compact dynamical horizons for the case \(c=0\) is established. Finally, an invariant characterization of this class of compact dynamical horizons is also presented. 
\end{abstract}

\maketitle

\section{Introduction}

A new covariant and gauge invariant way of studying black hole horizons \cite{rit1,shef1,shef2} has emerged over the last few years. Though used in a limited way to date, it provides a computationally inexpensive method for determining various geometric and thermodynamic properties of black hole horizons \cite{shef1,shef2}. This method employs the \(1+1+2\) semitetrad covariant splitting of spacetimes, which describes the spacetimes using well defined geometric and matter variables (see the references \cite{cc1,crb1,gbc1}). The first use of this approach in the study of black hole horizons, as far as we are aware, was carried out in 2014 by Ellis \textit{et al.}, \cite{rit1}, where the authors considered a gravitational collapse scenario in a realistic astrophysical setting, considering examples from the relatively small class of locally rotationally symmetric class II spacetimes. 

Ellis \textit{et al.} \cite{rit1} considered a case in an astrophysical setting where an initial marginally trapped surface, at the beginning of a gravitational collapse scenario, bifurcates into evolving surfaces, one being a timelike marginally trapped tube (to later be defined) that evolves inward and the other being a spacelike marginally trapped tube which evolves outward, expanding under the infalling of radiation, and approaches asymptotically a null marginally trapped tube. The causal character of the marginally trapped tubes (see the references \cite{ash1,ash2,ash3,ak1,ib1,boo2,ibb1,ib3} for more discussions on marginally trapped tubes) was determined by the slope of the tangent to the marginally trapped surfaces foliating them. Noting that  the covariant derivative of the outgoing null expansion scalar is normal to the marginally trapped tubes and marginally trapped surfaces, the dot product with the tangent to the marginally trapped surfaces vanishes, which allows one to determine the expression for the slope. In more general spacetimes this approach fails, or at best an explicit expression for the slope is not possible. This was the primary nature of two works by Sherif \textit{et al.} \cite{shef1,shef2} which were extensions of the work by Ellis \textit{et al.} \cite{rit1}, where two approaches were established; in one of the approaches the norm of the covariant derivative of the outgoing null expansion scalar was used, and in the other approach a smooth function on the marginally trapped tube - an approach by Booth and coauthors \cite{ibb1,ib3}, albeit restricted to spherically symmetric spacetimes - was expressed in terms of the \(1+1+2\) covariant variables. These approaches have allowed Sherif and coauthors to obtain well established results including the stability and instability of marginally trapped surfaces (see the following references \cite{and1,and2,yau1,jang1} for a discussion on this subject) of a Schwarzschild black hole and the Oppenheimer-Snyder dust collapse, as well as the bounds on the equation of state parameter determining the causal character of horizons in the Robertson-Walker spacetimes (these results were obtained by Ben-Dov \cite{bend1}). The third law of black hole thermodynamics was shown for locally rotationally symmetric class II spacetimes and a classification scheme was provided for diffeomorphically equivalent and causally equivalent marginally trapped tubes.

Thurston's \textit{geometrization conjecture} \cite{thu1,thu2,thu3} and Hamilton's proof of the \textit{uniformization theorem} \cite{ham1,ham2,ham3} are two of the most fundamental and important results in geometry and geometric analysis. These results provide the standard classification of \(2\)-dimensional and \(3\)-dimensional smooth manifolds. The \textit{Ricci flow} developed by Hamilton \cite{ham3} to study the evolution of metrics on smooth manifolds, provides a way of classifying the geometry of smooth closed \(3\)-manifolds. In particular, Hamilton showed that a closed smooth manifold admitting a metric of positive curvature is uniquely spherical. As a consequence, a compact Riemannian \(3\)-manifold admitting a positive metric cannot be foliated by hyperbolic planes. This result will be used in Section \ref{soc2} when we restrict the allowable geometry under the compactness conditions that will be imposed on horizons in locally rotationally symmetric class II spacetimes. The compactness condition can be obtained using the \textit{Bonnet-Myers Theorem} \cite{mye1} which gives a bound on the radius of a manifold with Ricci curvature strictly positive.

From a purely geometric viewpoint, finding specific examples of compact \(3\)-manifolds in spacetime to investigate both local and global geometric properties makes sense as there is a wealth of literature on \(n\)-dimensional compact Riemannian manifolds. The effectiveness of the adapting of horizon analysis to the \(1+1+2\) covariant variables - as demonstrated in the works by Sherif \textit{et al.} \cite{rit1,shef1,shef2} - coupled with standard results for compact \(3\)-manifolds should prove very useful in exposing the intricate balance between geometry and thermodynamics on black hole horizons. This work aims to identify certain classes of compact dynamical horizons in spacetimes of the Locally rotationally symmetric class (LRS II) using existing geometric analytic tools from Riemannian geometry, and investigate their geometric and thermodynamic properties. 

The paper is organized as follows: in section \ref{soc1} we briefly discuss the semi-tetrad covariant approach to be followed throughout this work, and then proceed to provide definitions needed to clarify the discourse of the paper. Black hole horizons, the required energy conditions, and additional properties are introduced in a covariant way. In section \ref{soc2} we provide a compactness theorem for dynamical horizons in LRS II spacetimes, evoking the well known Bonnet-Myers theorem. The properties of, and interplay between the geometry and thermodynamics of these obtained classes of compact horizons are then investigated using the Ricci flow evolution equation. The existence of solutions to the evolution equation and the geometric restrictions are also investigated. Finally, we conclude with a discussions of our results in section \ref{soc5}. 

\section{Preliminaries}\label{soc1}

In this section we provide a review of some background material on the \(1+1+2\) covariant splitting of LRS II spacetimes, as well as useful definitions so that the reader keeps track of concepts that will be used throughout the rest of the paper.

\subsection{\(1+1+2\) covariant splitting of LRS II spacetimes}

Any \(4\)-vector \(U^{\mu}\) in a spacetime manifold may be split into a component along a unit timelike vector field \(u^{\mu}\) and a component on the \(3\)-space as

\begin{eqnarray*}
U^{\mu}&=&Uu^{\mu} + U^{\langle \mu \rangle }.
\end{eqnarray*}
The scalar \(U\) is the scalar along \(u^{\mu}\) and \(U^{\langle \mu \rangle }\) is the projected \(3\)-vector \cite{ggff1,sge1} projected via the tensor \(h_{\mu}^{\ \nu}\equiv g_{\mu}^{\ \nu}+u_{\mu}u^{\nu}\). This \(1+3\) splitting irreducibly splits the covariant derivative of \(u^{\mu}\) as
\begin{eqnarray}\label{mmmn}
\nabla_{\mu}u_{\nu}=-A_{\mu}u_{\nu}+\frac{1}{3}h_{\mu\nu}\Theta+\sigma_{\mu\nu},
\end{eqnarray}
and the energy momentum tensor to be decomposed as
\begin{eqnarray}
T_{\mu\nu}=\rho u_{\mu}u_{\nu} + 2q_{(\mu}u_{\nu)} +ph_{\mu\nu} + \pi_{\mu\nu}.
\end{eqnarray}
The vector \(A_{\mu}=\dot{u}_{\mu}\) is the acceleration vector, \(\Theta\equiv D_{\mu}u^{\mu}\) - the trace of the fully orthogonally projected covariant derivative of \(u^{\mu}\) - is the expansion and \(\sigma_{\mu\nu}=D_{\langle \nu}u_{{\mu}\rangle}\) is the shear tensor. (Wherever used in this paper, angle brackets will denote the projected symmetric trace-free part of the tensor.) The quantity \(\rho\equiv T_{\mu\nu}u^{\mu}u^{\nu}\) is the energy density, \(q_{\mu}=-h_{\mu}^{\ \nu}T_{\nu\gamma}u^{\gamma}\) is the \(3\)-vector defining the heat flux, \(p\equiv\left(1/3\right)h^{\mu\nu}T_{\mu\nu}\) is the isotropic pressure and \(\pi_{\mu\nu}\) is the anisotropic stress tensor.

Whenever there is a preferred unit normal spatial direction \(e^{\mu}\) one may split the \(3\)-space into a direction along \(e^{\mu}\) and a \(2\)-surface where the projection tensor defined as

\begin{eqnarray}
N_{\mu\nu}=g_{\mu\nu}+u_{\mu}u_{\nu}-e_{\mu}e_{\nu}.
\end{eqnarray} 
The projection tensor \(N_{\mu\nu}\) projects any \(2\)-vector orthogonal to \(u^{\mu}\) and \(e^{\mu}\) onto the \(2\)-surface defined by the sheet (\(N^{\ \mu}_{\mu}=2\)). Thus \(u^{\mu}N_{\mu\nu}=0, \ e^{\mu}N_{\mu\nu}=0\). The vectors \(u^{\mu}\) and \(e^{\mu}\) are normalized so that \(u^{\mu}u_{\mu}=-1\) and \(e^{\mu}e_{\mu}=1\). This is referred to as the \(1+1+2\) splitting.

This splitting of the spacetime gives rise to four derivatives:
\begin{itemize}
\item For an arbitrary tensor \(S^{\mu..\nu}_{\ \ \ \ \gamma..\delta}\), one defines the \textit{covariant time derivative} (or simply the dot derivative)  along the observers' congruence of \(S^{\mu..\nu}_{\ \ \ \ \gamma..\delta}\) as \(\dot{S}^{\mu..\nu}_{\ \ \ \ \gamma..\delta}\equiv u^{\sigma}\nabla_{\sigma}S^{\mu..\nu}_{\ \ \ \ \gamma..\delta}\).

\item For an arbitrary tensor \(S^{\mu..\nu}_{\ \ \ \ \gamma..\delta}\) one defines the fully orthogonally \textit{projected covariant derivative} \(D\) with the tensor \(h_{\mu\nu}\) as \(D_\sigma S^{\mu..\nu}_{\ \ \ \ \gamma..\delta}\equiv h^{\mu}_{\ \rho}h^{\eta}_{\ \gamma}...h^{\nu}_{\ \tau}h^{\iota}_{\ \delta}h^{\lambda}_{\ \sigma}\nabla_{\lambda}S^{\rho..\tau}_{\ \ \ \ \eta..\iota}\).

\item Given a \(3\)-tensor \(\psi^{\mu..\nu}_{\ \ \ \ \gamma..\delta}\) the spatial derivative along the vector field \(e^{\mu}\) (simply called the \textit{hat derivative}) is given by \(\hat{\psi}_{\mu..\nu}^{\ \ \ \ \gamma..\delta}\equiv e^{\sigma}D_{\sigma}\psi_{\mu..\nu}^{\ \ \ \ \gamma..\delta}\).

\item Given a \(3\)-tensor \(\psi^{\mu..\nu}_{\ \ \ \ \gamma..\delta}\) the projected spatial derivative on the \(2\)-sheet (projection by the tensor \(N_{\mu}^{\ \nu}\)), called the \textit{delta derivative}, is given by \(\delta_\sigma\psi_{\mu..\nu}^{\ \ \ \ \gamma..\delta}\equiv N_{\mu}^{\ \rho}..N_{\nu}^{\ \tau}N_{\eta}^{\ \gamma}..N_\iota^{\ \delta}N_{\sigma}^{\ \lambda}D_{\lambda}\psi_{\rho..\tau}^{\ \ \ \ \eta..\iota}\).
\end{itemize}
Note that the projections by the tensors \(h^{\mu\nu}\) and \(N^{\mu\nu}\) in the definitions of the \(D\) and \(\delta\) derivatives are carried out over all indices (see the references \cite{cc1,pg1,ggff2} for more discussions).

\begin{definition}
A \textbf{locally rotationally symmetric class} II (LRS II) spacetime is an evolving, vorticity free and spatial twist free spacetime with a one dimensional isotropy group of spatial rotations defined at each point of the spacetime. It is given by the general line element 
\begin{eqnarray}\label{fork}
\begin{split}
ds^2&=-A^2\left(t,\chi\right)+ B^2\left(t,\chi\right)+ F^2\left(t,\chi\right) \left(dy^2+G^2\left(y,k\right)dz^2\right),
\end{split}
\end{eqnarray}
where \(t,\chi\) are parameters along integral curves of the timelike vector field \(u^{\mu}=A^{-1}\delta^{\mu}_0\) of a timelike congruence and the preferred spacelike vector \(e^{\mu}=B^{-1}\delta_{\nu}^{\mu}\) respectively. The constant \(k\) fixes the function \(G\left(y,k\right)\) (\(k=-1\) corresponds to \(\sinh y\), \(k=0\) corresponds to \(y\), \(k=1\) corresponds to \(\sin y\)) \cite{cc1,ggff1,ggff2}. 
\end{definition}

For LRS II spacetimes, all vector and tensor quantities vanish identically and the Weyl tensor is purely electric (see reference \cite{cc1} for details). Therefore the complete set of \(1+1+2\) covariant scalars fully describing the LRS class of spacetimes are 
\begin{eqnarray*}
\lbrace{A,\Theta,\phi, \Sigma, \mathcal{E}, \rho, p, \Pi, Q\rbrace}. 
\end{eqnarray*}
The quantity \(\phi\equiv\delta_{\mu}e^{\mu}\) is the sheet expansion, \(\Sigma\equiv\sigma_{\mu\nu}e^{\mu}e^{\nu}\) is the scalar associated with the shear tensor \(\sigma_{\mu\nu}\), \(\mathcal{E}\equiv E_{\mu\nu}e^{\mu}e^{\nu}\) is the scalar associated with the electric part of the Weyl tensor \(E_{\mu\nu}\), \(\Pi\equiv\pi_{\mu\nu}e^{\mu}e^{\nu}\) is the anisotropic stress scalar, and \(Q\equiv -e^{\mu}T_{\mu\nu}u^{\nu}=q_{\mu}e^{\mu}\) is the scalar associated to the heat flux vector \(q_{\mu}\).

The full covariant derivatives of the vector fields \(u^{\mu}\) and \(e^{\nu}\) are given by \cite{cc1}
\begin{subequations}\label{4}
\begin{align}
\nabla_{\mu}u_{\nu}&=-Au_{\mu}e_{\nu} + e_{\mu}e_{\nu}\left(\frac{1}{3}\Theta + \Sigma\right)+\frac{1}{2} N_{\mu\nu}\left(\frac{2}{3}\Theta -\Sigma\right),\label{4}\\
\nabla_{\mu}e_{\nu}&=-Au_{\mu}u_{\nu} + \left(\frac{1}{3}\Theta + \Sigma\right)e_{\mu}u_{\nu} +\frac{1}{2}\phi N_{\mu\nu}.\label{444}
\end{align}
\end{subequations}
We also note the useful expression 
\begin{eqnarray}\label{redpen}
\hat{u}^{\mu}&=&\left(\frac{1}{3}\Theta+\Sigma\right)e^{\mu}.
\end{eqnarray}

We will make use of the following commutation relation between the dot and hat derivatives when acting on an arbitrary scalar \(\psi\) in LRS II spacetimes:

\begin{eqnarray}\label{ghh1}
\hat{\dot{\psi}}-\hat{\dot{\psi}}=-A\dot{\psi}+\left(\frac{1}{3}\Theta+\Sigma\right)\hat{\psi}.
\end{eqnarray}
The field equations for LRS spacetimes are given as propagation and evolution of the covariant scalars \cite{cc1,rit1}:

\begin{itemize}

\item \textit{Evolution}
\begin{subequations}
\begin{align}
\frac{2}{3}\dot{\Theta}-\dot{\Sigma}&=A\phi- \frac{1}{2}\left(\frac{2}{3}\Theta-\Sigma\right)^2 + \mathcal{E} - \frac{1}{2}\Pi- \frac{1}{3}\left(\rho+3p\right),\label{subbe1}\\
\dot{\phi}&=\left(\frac{2}{3}\Theta-\Sigma\right)\left(A-\frac{1}{2}\phi\right) + Q,\label{subbe2}\\
\dot{\mathcal{E}}-\frac{1}{3}\dot{\rho}+\frac{1}{2}\dot{\Pi}&=-\left(\frac{2}{3}\Theta-\Sigma\right)\left(\frac{3}{2}\mathcal{E}+\frac{1}{4}\Pi\right)+\frac{1}{2}\phi Q+\frac{1}{2}\left(\frac{2}{3}\Theta-\Sigma\right)\left(\rho+p\right),\label{subbe3}
\end{align}
\end{subequations}
\item \textit{Propagation}
\begin{subequations}
\begin{align}
\frac{2}{3}\hat{\Theta}-\hat{\Sigma}&=\frac{3}{2}\phi\Sigma + Q,\label{subbe4}\\
\hat{\phi}&=-\frac{1}{2}\phi^2 + \left(\frac{1}{3}\Theta+\Sigma\right)\left(\frac{2}{3}\Theta-\Sigma\right)-\frac{2}{3}\rho-\mathcal{E}-\frac{1}{2}\Pi,\label{subbe5}\\
\hat{\mathcal{E}}-\frac{1}{3}\hat{\rho}+\frac{1}{2}\hat{\Pi}&=-\frac{3}{2}\phi\left(\mathcal{E}+\frac{1}{2}\Pi\right)-\frac{1}{2}\left(\frac{2}{3}\Theta-\Sigma\right)Q\label{subbe6}
\end{align}
\end{subequations}
\item \textit{Evolution/Propagation}
\begin{subequations}
\begin{align}
\hat{A}-\dot{\Theta}&=-\left(A+\phi\right)A-\frac{1}{3}\Theta^2+\frac{3}{2}\Sigma^2+\frac{1}{2}\left(\rho+3p\right),\label{subbe7}\\
\dot{\rho}+\hat{Q}&=-\Theta\left(\rho+p\right)-\left(2A+\phi\right)Q-\frac{3}{2}\Sigma\Pi,\label{subbe8}\\
\dot{Q}+\hat{p}+\hat{\Pi}&=-\left(A+\frac{3}{2}\phi\right)\Pi-\left(\frac{4}{3}\Theta+\Sigma\right)Q-\left(\rho+p\right)A.\label{subbe9}
\end{align}
\end{subequations}

\end{itemize}

\subsection{Some useful definitions}

We will now give some definitions used in describing black hole spacetimes and associated horizons. What are to follow are all very familiar definitions and we will follow mostly standard references \cite{ash1,ash2,ash3,boo2,ibb1}.

Given an embedded \(2\)-manifold \(S\subset M\), one may define two normal vector fields \(k^{\mu}\) and \(l^{\mu}\), called the outgoing and ingoing null normal vector fields associated with outgoing and ingoing null geodesics. The vector fields \(k^{\mu}\) and \(l^{\mu}\) are normalized to satisfy the relations

\begin{eqnarray*}
k^{\mu}k_{\mu}=l^{\mu}l_{\mu}=0;\ \ \ k^{\mu}l_{\mu}=-1.
\end{eqnarray*} 
Associated with \(S\) are functions, defined for each of the null normal directions which we denote \(\Theta_k\) and \(\Theta_l\) respectively. These are called the null normal expansions in the \(k^{\mu}\) and \(l^{\mu}\) directions.

\begin{definition}[\textbf{Marginally trapped surface (MTS)}]\label{def5}
An embedded \(2\)-surface \(S\) in \(M\) is said to be marginally trapped if for all points of \(S\), \(\Theta_k=0\) and \(\Theta_l\in\mathbb{R}^-\) (where \(\mathbb{R}^-\) denotes the set of negative real numbers).
\end{definition}
These \(2\)-surfaces foliate hypersurfaces in spacetime that, under certain conditions, may be associated to the boundary of a black hole. The notion of these hypersurfaces which generalizes Hayward's \textit{future outer trapping horizon} (FOTH) was introduced by Ashtekar and Galloway \cite{ash3}: 

\begin{definition}[\textbf{Marginally trapped tube (MTT)}]\label{def6}
A marginally trapped tube is a codimension \(1\) embedded submanifold in a spacetime foliated by marginally trapped surfaces.
\end{definition}

In general, the sign of the induced metric on a marginally trapped tube may vary. In specific cases where the sign of the metric is not changing as one moves along the marginally trapped tube, the marginally trapped tube is called a timelike membrane (TLM), a non-expanding horizon (NEH), or a dynamical horizon depending on the sign of the metric which depends of the formalism used (we will elaborate on this shortly).

\begin{figure}
\def\svgwidth{12cm}
\def\svgheight{7cm}
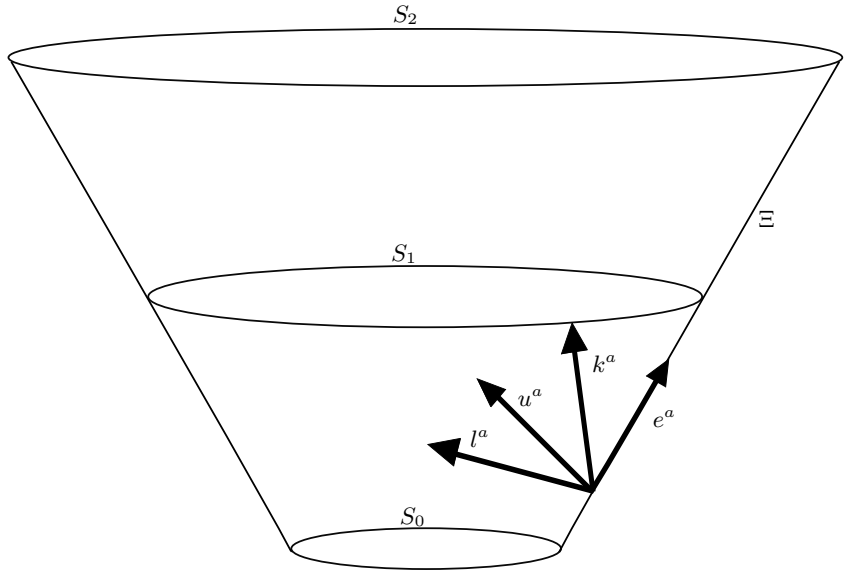
\caption{The usual depiction of a dynamical horizon \(\Xi\) foliated by marginally trapped surfaces S, showing the null normal vector fields \(k^a\) and \(l^a\), as well as the unit timelike and unit normal vector fields \(u^a\) and \(e^a\) respectively.}
\label{fig:a}
\end{figure}
Above we provide a depiction of a dynamical horizon in figure \ref{fig:a} (original depictions appearing in  \cite{ash1,ash2,ak1}), where it is pictured as a hyperboloid in Minkowski space. In this schematic, motions along the unit normal vector field \(e^a\) are interpreted as time evolution with respect to observers at infinity. Similar picture of a timelike membrane as a hyperboloid in Minkowski space can be presented, where in this case, there is a decrease in the surface area of the ``discs" depicting the marginally trapped surfaces \(S_i\) in figure \ref{fig:a}, as one moves along \(e^a\). The case of an non-expanding horizon can be depicted as a cylinder in Minkowski space. In each case the marginally trapped surfaces are intersections of the hyperboloid (cylinder) with spacelike planes.

\subsection{Marginal trapping and marginally trapped tubes in LRS II spacetimes}

For LRS II spacetimes the null expansion scalars associated with the outgoing and ingoing null normal vector fields to \(2\)-surfaces in a spacetime are given by linear combinations of the shear, expansion and sheet expansion covariant scalars. Explicitly these scalars are given by \cite{rit1,shef1,shef2}

\begin{subequations}
\begin{align}
\Theta_k&=\frac{1}{\sqrt{2}}\left(\frac{2}{3}\Theta-\Sigma+\phi\right),\label{subb1}\\
\Theta_l&=\frac{1}{\sqrt{2}}\left(\frac{2}{3}\Theta-\Sigma-\phi\right).\label{subb2}
\end{align}
\end{subequations}
The requirement that a marginally trapped surface satisfies \(\Theta_k=0\) and \(\Theta_l<0\) implies that we must have \(\phi>0\), which implies that \(\left(2/3\right)\Theta-\Sigma<0\). This then restricts the subclass of LRS II spacetimes potentially admitting marginally trapped surfaces. From \eqref{subb1} and \eqref{subb2}, it is clear that whenever \(\phi=0\) we must have \(\left(2/3\right)\Theta-\Sigma=0\) in which case the \(2\)-surface is minimal (see references \cite{shef1,shef2} for further discussion). Thus we have ruled out minimal surfaces in LRS II spacetimes for the rest of this paper.

As was discussed in the previous subsection, marginally trapped tubes are foliated by marginally trapped surfaces. In the case the metric signature on a marginally trapped tube is fixed, then it may be classed as timelike, non-expanding or spacelike. 

Determining the metric signature on an marginally trapped tube is not unique. However, one can always relate the different formulations that compute the metric signature. Two explicit formulations have been utilized by applying the \(1+1+2\) covariant formalism. The first is a formalism developed in reference \cite{shef1}. Briefly put, a choice of vector field, dependent on some smooth function - this function being denoted \(C\) - is made, which is tangent to the marginally trapped tube and everywhere orthogonal to the foliation (this approach was developed by Booth and coauthors (see reference \cite{ib3} and associated references) and utilized for the well known spherically symmetric spacetimes, but was generalized to all of a more diversed class of spacetimes and interpreted in terms of the \(1+1+2\) covariant variables):

\begin{eqnarray}\label{gghh1}
\mathcal{V}^{\mu}=k^{\mu}-Cl^{\mu}.
\end{eqnarray}
The definition of \(\mathcal{V}^{\mu}\) implies the outgoing null expansion scalar remains fixed as it is Lie dragged along \(\mathcal{V}^{\mu}\), which immediately gives

\begin{eqnarray}\label{gghh2}
C=\frac{\mathcal{L}_k\Theta_k}{\mathcal{L}_l\Theta_k},
\end{eqnarray}
where \(\mathcal{L}_n\) denotes the Lie derivative along the vector field \(n^{\mu}\). If \(C<0\), \(C=0\) or \(C>0\), and \(C\) is such that it is fixed all over the marginally trapped tube, then the marginally trapped tube is a timelike membrane, a non-expanding horizon or a dynamical horizon. A second approach (which we will not take into account in this paper but which may however be related to \(C\)) notes that the gradient of \(\Theta_k\), \(\nabla_{\mu}\Theta_k\), is normal to the marginally trapped tube and as such the sign of the norm of \(\nabla_{\mu}\Theta_k\) can be used to determine the causal character of the marginally trapped tube. In fact it can easily be seen that we can write (the references \cite{shef1,shef2} have the description of the procedures)

\begin{eqnarray}\label{gghh3}
C^*=\nabla_{\mu}\Theta_k\nabla^{\mu}\Theta_k=-\mathcal{L}_k\Theta_k\mathcal{L}_l\Theta_k,
\end{eqnarray}
which allows us to write

\begin{eqnarray}\label{gghh4}
C=-\frac{1}{\left(\mathcal{L}_l\Theta_k\right)^2}C^*.
\end{eqnarray}
Of course then the signs of \(C\) and \(C^*\) are reversed. In this case if \(C^*>0\) the marginally trapped tube is timelike and if \(C^*<0\) then the marginally trapped tube is spacelike. The requirement that the null energy condition (NEC) be satisfied on a marginally trapped tube (which we will assume throughout the rest of the paper) is equivalent to the condition that \(\mathcal{L}_k\Theta_k\leq 0\). 

This paper focuses on dynamical horizons exclusively and as such we will specify basic conditions that will necessarily satisfied. For LRS II spacetimes \(C\) is explicitly calculated as \cite{shef1}

\begin{eqnarray}\label{gghh5}
C=\frac{-\left(\rho+p+\Pi\right)+2Q}{\frac{1}{3}\left(\rho-3p\right)+2\mathcal{E}}.
\end{eqnarray}
Clearly if \(\mathcal{L}_k\Theta_k= 0\) then \(C=0\), a n0n-expanding horizon. Therefore, if we are assuming the null energy condition is satisfied then this amounts to the energy condition \(\rho+p+\Pi>2Q\). As \(C\) must be greater than zero on the dynamical horizon we must therefore have \(\rho<p-6\mathcal{E}\). Thus we obtain the required energy condition on a dynamical horizon in LRS II spacetimes: \(2Q-\Pi<\left(\rho+p\right)<2\left(p-3\mathcal{E}\right)\) (see the reference \cite{shef1} for further details).

There are certain cases that may be immediately ruled out, i.e. certain subclass of LRS II spacetimes can be determined to not admit a dynamical horizon. One of them is the shear-free case. As has been shown by Sherif \textit{et al.} \cite{shef2}, for a dynamical horizon the expansion \(\Theta\) is strictly positive. It is clear then that the shear-free case can admit no dynamical horizon. This is because in the shear-free case one has from the vanishing of \eqref{subb1} \(\left(2/3\right)\Theta=-\phi\), and since \(\phi>0\) we must have \(\Theta<0\). This in fact clearly shows that any marginally trapped tube in a shear-free LRS II spacetime will necessarily be a timelike membrane \cite{shef2}.

In this work we shall specialize to dynamical horizons. Throughout this paper we shall simply write \textit{horizon} whenever we are referring to a dynamical horizon. We shall also assume the \textit{genericity condition} of Ashtekar \& Galloway \cite{ash3}, i.e. 

\begin{eqnarray}\label{genco}
\sigma_{\mu\nu}\sigma^{\mu\nu}+T_{\mu\nu}k^{\mu}k^{\nu}\neq 0,
\end{eqnarray} 
holds true on the horizon, which, in the case of LRS II spacetimes, translates to the condition that the weak energy condition (WEC) is either strictly satisfied or otherwise violated on the horizon (\(\rho+p>0\) or \(\rho+p<0\)).

\section{Compact dynamical horizons in LRS II spacetimes}\label{soc2}

For LRS II spacetimes, with both the unit vectors \(u^{\mu}\) and \(e^{\mu}\) being hypersurface orthogonal, the Ricci tensor for any spacelike \(3\)-surface is given by \cite{cc1}

\begin{eqnarray}\label{go1}
\begin{split}
R_{\mu\nu}&=-\left(\hat{\phi}+\frac{1}{2}\phi^2\right)e_{\mu}e_{\nu}-\left[\frac{1}{2}\left(\hat{\phi}+\phi^2\right)-K\right]N_{\mu\nu},
\end{split}
\end{eqnarray}
where \(K\) is the Gaussian curvature of the \(2\)-sheet which is given by \cite{cc1} 

\begin{eqnarray}\label{go2}
K=\frac{1}{3}\rho-\mathcal{E}-\frac{1}{2}\Pi+\frac{1}{4}\phi^2-\frac{1}{4}\left(\frac{2}{3}\Theta-\Sigma\right)^2,
\end{eqnarray}
whose dot and hat derivatives are respectively given by

\begin{subequations}
\begin{align}
\dot{K}&=-\left(\frac{2}{3}\Theta-\Sigma\right)K\label{gau1}\\
\hat{K}&=\phi K.\label{gau2}
\end{align}
\end{subequations}
The Ricci scalar on the \(3\)-manifold is given by

\begin{eqnarray}\label{go79}
R=-2\left(\hat{\phi}+\frac{3}{4}\phi^2-K\right).
\end{eqnarray}
In all that is to follow we will usually set 

\begin{eqnarray*}
\alpha=-\left(\hat{\phi}+\frac{1}{2}\phi^2\right)\ \ \text{and}\ \ \beta=-\left[\frac{1}{2}\left(\hat{\phi}+\phi^2\right)-K\right].
\end{eqnarray*}
We will emphasize that \(R\) without an index specifying the space we are working in, we will always be referring to embedded \(3\)-submanifolds. Whenever we are referring to the ambient spacetime or the marginally trapped surfaces foliating the \(3\)-manifolds \(R\) will be specifically indexed for that purpose. In the case that the spacelike \(3\)-manifold is foliated by \(2\)-surfaces that are marginally trapped, the Ricci tensor can be written as

\begin{eqnarray}\label{go80}
\begin{split}
R_{\mu\nu}&=\left[\frac{2}{3}\rho+\left(\mathcal{E}+\frac{1}{2}\Pi\right)-\phi\left(\Theta+\phi\right)\right]e_{\mu}e_{\nu}+\left[\frac{2}{3}\rho-\frac{1}{2}\left(\mathcal{E}+\frac{1}{2}\Pi\right)-\frac{3}{4}\phi\left(\Theta+\phi\right)\right]N_{\mu\nu},
\end{split}
\end{eqnarray}
where we have used \eqref{subbe5} and the vanishing of \(\Theta_k\). Thus the Ricci scalar becomes simply

\begin{eqnarray}\label{go81}
R=2\rho-\frac{5}{2}\phi\left(\Theta+\phi\right).
\end{eqnarray}

Throughout we will assume that the Ricci tensor on the \(3\)-manifolds do not vanish. Consequently we have the restriction 

\begin{eqnarray*}
\alpha,\beta\neq 0;\ \ R\neq 0
\end{eqnarray*}
on the horizon, which combines to give

\begin{eqnarray}\label{dff1}
\frac{2}{5}\rho\neq -3\left(\mathcal{E}+\frac{1}{2}\Pi\right).
\end{eqnarray}

The relationship between the geometry and the Ricci curvature of a Riemannian manifold has been extensively studied (see the references \cite{ham3,yano1}). In particular, bounds on the Ricci curvature have shed many insights on topological properties of Riemannian manifolds \cite{mye1}. If these manifolds are foliated by marginally trapped surfaces, what general properties, both geometric and topological, can be obtained?

We state and prove a compactness result for embedded \(3\)-manifolds in LRS II spacetimes, which depends on a first order spatial differential equation in the sheet expansion.

\begin{theorem}\label{th2}
Let \(M\) be an LRS II spacetime and \(\Xi\) an embedded spacelike \(3\)-manifold with \(R_{\mu\nu}\neq 0\). Furthermore, suppose
\begin{eqnarray}\label{theq1}
\hat{\phi}+\frac{3}{4}\phi^2-cK=0
\end{eqnarray}
is satisfied at all points of \(\Xi\). If on the marginally trapped surfaces foliating \(\Xi\) we have \(K>0\) with \(c=0\) or \(K<0\) with \(c=2\) for finite \(K\), then \(\Xi\) is compact.  
\end{theorem}

\begin{proof}
The mode of the proof will be to invoke the well known \textit{Bonnet-Myers theorem} \cite{mye1} which, simply put, implies compactness of an \(n\)-dimensional Riemannian manifold \(M\) if its Ricci curvature \(R\) is bounded below by 

\begin{eqnarray}\label{micheila1}
\left(n-1\right)m>0,
\end{eqnarray}
for some constant \(m\).

First it is clear that for \(K>0\) we have (from \eqref{gau1},\eqref{gau2}) \(\dot{K},\hat{K}>0\), so \(K\) is positive all over \(\Xi\). We may rewrite \eqref{go79} as

\begin{eqnarray}\label{theq2}
R=\left[-\frac{2}{K}\left(\hat{\phi}+\frac{3}{4}\phi^2\right)+2\right]K,
\end{eqnarray}
so that \(R\) assumes the form of the left hand side of \eqref{micheila1}, with \(m=K\). The bracketized term of \eqref{theq2} can now be equated to \(\left(n-1\right)\) as (for \(n=3\)):

\begin{eqnarray}\label{theq3}
\left(n-1\right)=2=-\frac{2}{K}\left(\hat{\phi}+\frac{3}{4}\phi^2\right)+2,
\end{eqnarray}
so that

\begin{eqnarray}\label{theq4}
\hat{\phi}+\frac{3}{4}\phi^2=0.
\end{eqnarray}
Then if \(K>0\) we will have \(\left(n-1\right)K=2K>0\). Similarly, we can set

\begin{eqnarray}\label{theq5}
-\left(n-1\right)=-2=-\frac{2}{K}\left(\hat{\phi}+\frac{3}{4}\phi^2\right)+2,
\end{eqnarray}
so that

\begin{eqnarray}\label{theq6}
\hat{\phi}+\frac{3}{4}\phi^2-2K=0.
\end{eqnarray}
Then if \(K<0\) we will have \(\left(n-1\right)\left(-K\right)=-2K>0\). In either case we have \(R>0\) with \(R=\pm 2K\)(`\(+\)' is for the case of \(K>0\) and `\(-\)' for the case of \(K<0\)). Thus we have \(\Xi\) being compact.\qed

\end{proof}

Notice that this result holds in general for embedded \(3\) manifolds in LRS II spacetimes and not just marginally trapped tubes, though this work restricts tomarginally trapped tubes. Here we note the implication of Theorem \ref{th2} that \textit{given an LRS II spacetime \(M\) and a compact \(3\)-surface \(\Xi\subset M\), if \(\hat{\phi}+\left(3/4\right)\phi^2-cK=0\) is satisfied on \(\Xi\) for \(c\in \lbrace{0,2\rbrace}\) - where \(K\) is the Gaussian curvature of \(2\)-surfaces \(S\) in \(M\), then the smooth embedding \(\varphi:S\longrightarrow \Xi\) of \(S\) into \(\Xi\) preserves the Ricci curvature}.

Notice here that we do not specify the geometry of the marginally trapped surfaces. While we are stating that we are considering the cases for \(K>0\) and \(K<0\) on the marginally trapped surfaces, we may interchangeably specify that the marginally trapped surfaces are spherical or hyperbolic in the respective cases. The Gaussian curvature \(K\) may vary over the marginally trapped surface but the sign is required to be fixed.

Our interest throughout this work will be to study various properties of the compact horizon types of Theorem \ref{th2}. Notice that as long as the kernel of the left hand side of \eqref{theq1} is nonempty (here we are viewing the left hand side of \eqref{theq1} as a function \(\Phi:\Xi\longrightarrow \mathbb{R}\)), then there is always a subset of \(\Xi\) that is compact, and all results that are to follow will hold as well on such subsets.

\subsection{Geometry of compact dynamical  horizons in LRS II spacetimes}

As a dynamical horizon evolves its geometry and topology may change. A consequence of this is that the thermodynamic quantities also evolve. The relationship between the geometry and thermodynamics may be investigated through analysis of the time evolution of the \(3\)-metric on the horizon. The condition in \eqref{theq1} provides a constraint on the subclass of spacetimes admitting compact horizons foliated by marginally trapped surfaces on which \(K>0\) or \(K<0\). There is an obstruction to the existence of the compact case with \(K<0\), and this will be shown. As the compactness we have determined is governed by differential equations of geometric quantities on \(\Xi\), we will now study properties of these horizons and in effect how the geometric evolution constrains the geometry and thermodynamics of \(\Xi\).

\subsubsection{Compact dynamical horizons in LRS II spacetimes foliated by marginally trapped surfaces with \(K>0\)}

Let us start with the case of \(c=0\), where \(K>0\) on the MTSs. Comparing \eqref{theq1} for \(c=0\) and \eqref{subbe5} we obtain

\begin{eqnarray}\label{theq8}
-\frac{2}{3}\rho-\mathcal{E}-\frac{1}{2}\Pi+\left(\frac{1}{3}\Theta+\Sigma\right)\left(\frac{2}{3}\Theta-\Sigma\right)=0,
\end{eqnarray}
which, on the horizon, simplifies \(K\) as 
\begin{eqnarray}\label{theq9}
\begin{split}
K&=\rho-\left(\frac{2}{3}\Theta-\Sigma\right)\left(\frac{4}{3}\Theta+\frac{5}{4}\phi\right)\\
&=\rho+\phi\left(\frac{4}{3}\Theta+\frac{5}{4}\phi\right).
\end{split}
\end{eqnarray}
Therefore since we must have \(\phi,\Theta>0\) (for a DH \(\Theta>0\)), it is sufficient to specify the energy density as positive on \(\Xi\) (it may be negative away from \(\Xi\)), though it is not necessary, so that \(K\) is always positive. 

\begin{proposition}\label{th3} 
Let \(M\) be an LRS II spacetime and let \(\Xi\) be a compact horizon in \(M\) foliated by marginally trapped \(2\)-surfaces with \(K>0\), satisfying \eqref{theq1} for \(c=0\). If the induced metric \(h_{\mu\nu}\) on \(\Xi\) is time dependent, then at some time \(T\in \left(a,t_{end}\right]\) for \(a>t_0\) (where \(t_0\) is the initial time of the metric evolution and \(t_{end}<\infty)\), there exists a metric on \(\Xi\) for which \(\Xi\) neither absorbs nor emits radiation.
\end{proposition}

\begin{proof}
We shall utilize the Ricci flow geometric evolution equation \cite{ham3} and show that the constraints generated by the flow implies \(Q\) must vanish. We will assume that the time coordinate parametrizes the family of metric on \(\Xi\).

The covariant time derivative of the metric on \(\Xi\) is given by

\begin{eqnarray}\label{theq10}
\begin{split}
u^{\delta}\nabla_{\delta}h_{\mu\nu}&=u^{\delta}\nabla_{\delta}\left(g_{\mu\nu}+u_{\mu}u_{\nu}\right)\\
&=u^{\delta}\nabla_{\delta}\left(u_{\mu}u_{\nu}\right)\\
&=u_{(\mu}\dot{u}_{\nu)}\\
&=2Ae_{(\mu}u_{\nu)}.
\end{split}
\end{eqnarray}
For a compact \(3\)-manifold (Riemannian) the Ricci flow equation is normalized as 

\begin{eqnarray}\label{theq11}
\begin{split}
u^{\delta}\nabla_{\delta}h_{\mu\nu}=-2R_{\mu\nu}+\frac{2}{3}Rh_{\mu\nu},
\end{split}
\end{eqnarray}
which for LRS II spacetimes can explicitly be written as

\begin{eqnarray}\label{theq12}
\begin{split}
Ae_{(\mu}u_{\nu)}=\frac{1}{3}\left(\beta-\alpha\right)\left(2e_{\mu}e_{\nu}-N_{\mu\nu}\right),
\end{split}
\end{eqnarray}
where the round brackets on the indices denote symmetrization. Contracting \eqref{theq12} by \(u^{\mu}u^{\nu},e^{\mu}e^{\nu},u^{(\mu}e^{\nu)}\) and \(N_{\mu\nu}\) we obtain the two independent equations

\begin{subequations}
\begin{align}
A&=0,\label{sube415}\\
\left(\beta-\alpha\right)&=0,\label{sube416}
\end{align}
\end{subequations}
in which case \eqref{theq12} is satisfied. Thus, as the metric evolves, \(A\) tends to zero and \(\alpha\) tends to \(\beta\), and this occurs at time 

\begin{eqnarray}\label{theq13}
\begin{split}
t=T\vline_{(A=0,\alpha=\beta)}.
\end{split}
\end{eqnarray}
The condition \(\alpha=\beta\) implies that \(K=-\left(1/2\right)\hat{\phi}\), which using \eqref{theq1} for \(c=0\) gives

\begin{eqnarray}\label{theq14}
K=\frac{3}{8}\phi^2.
\end{eqnarray}
Taking the dot derivative of \eqref{theq14} and comparing to \eqref{gau1} we obtain

\begin{eqnarray}\label{theq15}
\frac{3}{4}\phi\dot{\phi}=-\left(\frac{2}{3}\Theta-\Sigma\right)K,
\end{eqnarray}
which, upon inserting \eqref{subbe2} and noting that \(A=0\), yields a constraint on \(Q\)

\begin{eqnarray}\label{theq16}
\frac{3}{4}\phi Q=-\left(\frac{2}{3}\Theta-\Sigma\right)\left(K-\frac{3}{8}\phi^2\right).
\end{eqnarray}
On the horizon \eqref{theq16} simplifies as (using \eqref{theq14} to substitute for \(K\))

\begin{eqnarray}\label{theq17}
\phi Q=0,
\end{eqnarray}
in which case we must have either \(\phi=0\) or \(Q=0\) on \(\Xi\). If \(\phi=0\), then, from \eqref{theq14}, one has \(K=0\) and consequently \(R=0\). It is well known \cite{ham1} that, for a compact 3-manifold, if \(R>0\) for an initial metric (\(R=2K>0\) for the initial metric in the case considered here), then it holds true for all times \(t\), and hence we can rule out the case \(\phi=0\). Therefore we have that \(Q=0\) and the result follows.\qed
\end{proof}

We see in this case that \(K\) is always positive from \eqref{theq14}. In fact it can also be shown that for times \(T\neq T\vline_{(A=0,\alpha=\beta)}\), we have the following estimate for \(K\):

\begin{eqnarray}\label{happy}
K\geq \frac{1}{8}\phi^2.
\end{eqnarray}
To see this, we recall \cite{ham1} that, for a \(3\)-manifold with positive scalar curvature (or Ricci tensor), one has the estimate

\begin{eqnarray}\label{happy1}
\frac{1}{3}R^2\leq |R_{\mu\nu}|^2\leq R^2.
\end{eqnarray}
Using \(R_{\mu\nu}=\alpha e_{\mu}e_{\nu}+\beta N_{\mu\nu}\), we write \eqref{happy1} explicitly as

\begin{eqnarray}\label{happy3}
\frac{1}{3}\left(\alpha+2\beta\right)^2\leq \alpha^2+2\beta^2\leq \left(\alpha+2\beta\right)^2,
\end{eqnarray}
which can be split as

\begin{subequations}
\begin{align}
\frac{1}{3}\left(\alpha+2\beta\right)^2&\leq \alpha^2+2\beta^2\notag\\
\implies -\frac{2}{3}\left(\alpha-\beta\right)^2&\leq 0,\label{happy4}\\
\alpha^2+2\beta^2&\leq \left(\alpha+2\beta\right)^2\notag\\
\implies -2\beta R&\leq 0.\label{happy5}
\end{align}
\end{subequations}
Notice that \eqref{happy4} always holds. For \eqref{happy5} to hold, since \(R>0\), we must have \(\beta\geq 0\). Explicitly, noting that \(\hat{\phi}+(3/4)\phi^2=0\), we may write \(\beta\) as

\begin{eqnarray}\label{happy6}
\beta=K-\frac{1}{8}\phi^2,
\end{eqnarray}
and hence the result follows. As a consequence we have that 

\begin{eqnarray}\label{happy6}
R\geq \frac{1}{4}\phi^2.
\end{eqnarray} 
(Notice that \eqref{theq14} satisfies the estimate \eqref{happy}). Indeed, it makes sense intuitively that the sheet expansion controls the `growth' of the curvature \(R\), since the Gaussian curvature of the marginally trapped surfaces determines \(R\).

Hypothetically, consider this case: Let us consider a scenario where the horizon \(\Xi\) evolves along \(u^a\), so that at each time \(t\) of the horizon evolution we have an associated metric, a solution to \eqref{theq11}. Then Proposition \ref{th3} presents a situation where it is possible that \(i).\) the horizon may radiate for some time, after which it stops radiating, \(ii).\) the horizon has been non-radiating since its formation (we are assuming here a formation from an astrophysical collapse) or, \(iii).\) the horizon is initially absorbing radiation and after time \(T\vline_{(A=0,\alpha=\beta)}\) it stops absorbing radiation. From \eqref{theq14} \(K\) stays positive throughout the evolution of the metric on \(\Xi\), and so the geometry of the foliation is fixed. Insight into the general properties of these horizon types would require a thorough analysis to check consistency of the field equations on these horizons. For example, without a detailed and careful analysis of the field equations on these horizon types, one might wrongly assert that during the evolution of the metric, one goes from a positive definite metric to a negative definite one. To see this, we recall from \cite{ib3} that if \(\tilde{\epsilon}\) denotes the area form on the \(2\)-surfaces, then Lie dragging \(\tilde{\epsilon}\) along \(\mathcal{V}\) gives

\begin{eqnarray}\label{theq200}
\mathcal{L}_{\mathcal{V}}\tilde{\epsilon}=-C\Theta_l\tilde{\epsilon},
\end{eqnarray} 
so that expansion and contraction of an marginally trapped tube is in essence determined by the metric signature on the marginally trapped tube (noting that \(\Theta_l<0\)): A marginally trapped tube is timelike (\(C<0\)) if and only if it contracts (\(\Theta<0\)) and spacelike (\(C>0\)) if and only if it expands (\(\Theta>0\)). Let us consider times \(t\geq T\vline_{(A=0,\alpha=\beta)}\) when \(Q\) vanishes. Take the dot derivative of \eqref{theq8} and use \eqref{subbe1}, \eqref{subbe3} and \eqref{subbe8}. After some simplification, the resulting equation on the horizon simplifies to 

\begin{eqnarray}\label{theq200}
\Theta\phi\left(2\Theta+\frac{1}{2}\phi\right)=0.
\end{eqnarray} 
Since \(\Theta\phi\) cannot be zero on the horizon we must have \(\Theta=-\left(1/4\right)\phi\), and noting that \(\phi>0\) results in the conclusion that \(\Theta<0\) on the horizon.

The problem with this is that \(\Theta\) is smooth, and therefore one expects that the transition has to go through a \textit{non-expanding} phase if \(\Theta\) is to become negative, i.e. to occur as \(\Theta\to 0 \to \mathbb{R}^-\). We shall see that in fact the metric does become degenerate at time \(T\vline_{(A=0,\alpha=\beta)}\), i.e. \(\Theta=0\). 

\begin{proposition}\label{th3} 
Let \(M\) be an LRS II spacetime and let \(\Xi\) be a compact horizon in \(M\) foliated by marginally trapped \(2\)-surfaces with \(K>0\), satisfying \eqref{theq1}, and let the induced metric \(h_{\mu\nu}\) on \(\Xi\) is time dependent. Then \eqref{theq11} admits no solutions for time \(t= T\vline_{(A=0,\alpha=\beta)}\).
\end{proposition}

\begin{proof}
We will proceed with the proof by showing that as the induced metric is evolved \(\Xi\) becomes null for time \(t= T\vline_{(A=0,\alpha=\beta)}\), i.e. \(\Theta=0\). In this case we shall show that either \(\phi=0\) (this was ruled out on grounds that the case \(\phi=0\implies K=0\implies R=0\) which is not possible), or the shear scalar \(\Sigma\) is complex valued, or the strong energy condition has to be violated in which case it can be shown that the weak energy condition has to be violated or otherwise the isotropic pressure is negative (here we are assuming the generecity condition in which case \(\rho+p\neq 0\)). We apply the commutation relation in \eqref{ghh1}, on the pairs of evolution and propagation equations \eqref{subbe1} and \eqref{subbe4}, \eqref{subbe2} and \eqref{subbe5} and \eqref{subbe3} and \eqref{subbe6}. Taking the hat and dot derivatives of  \eqref{subbe1} and \eqref{subbe4} we obtain respectively (after simplifications)

\begin{eqnarray}\label{subeas1}
\begin{split}
\frac{2}{3}\hat{\dot{\Theta}}-\hat{\dot{\Sigma}}&=-\frac{3}{2}\phi\left[\Sigma\left(\frac{2}{3}\Theta-\Sigma\right)+\mathcal{E}+\frac{1}{2}\Pi\right]-\left(\hat{p}+\hat{\Pi}\right)\\
&=\frac{3}{2}\phi\left[-\Sigma\left(\frac{2}{3}\Theta-\Sigma\right)-\mathcal{E}+\frac{1}{2}\Pi\right],
\end{split}
\end{eqnarray}
and

\begin{eqnarray}\label{subeas4}
\frac{2}{3}\dot{\hat{\Theta}}-\dot{\hat{\Sigma}}=\frac{3}{2}\phi\left(\frac{4}{9}\Theta^2-\Theta\Sigma-\mathcal{E}+\frac{1}{2}\Pi\right).
\end{eqnarray}
Subtracting \eqref{subeas4} from \eqref{subeas1} we obtain

\begin{eqnarray}\label{subeas7}
\begin{split}
\left(\frac{2}{3}\hat{\dot{\Theta}}-\hat{\dot{\Sigma}}\right)-\left(\frac{2}{3}\dot{\hat{\Theta}}-\dot{\hat{\Sigma}}\right)&=-\frac{3}{2}\phi\left(\frac{4}{9}\Theta^2-\frac{1}{3}\Theta\Sigma-\Sigma^2\right).
\end{split}
\end{eqnarray}
Now, using the commutation relation \eqref{ghh1} we have

\begin{eqnarray}\label{subeas10}
\left(\frac{2}{3}\hat{\dot{\Theta}}-\hat{\dot{\Sigma}}\right)-\left(\frac{2}{3}\dot{\hat{\Theta}}-\dot{\hat{\Sigma}}\right)=\frac{3}{2}\phi\left(\frac{1}{3}\Theta\Sigma+\Sigma^2\right).
\end{eqnarray}
Comparing \eqref{subeas7} and \eqref{subeas10} we have the constraint

\begin{eqnarray}\label{subeas13}
0=\phi\Theta^2.
\end{eqnarray}
From \eqref{subeas13} we must have \(\Theta=0\) or \(\phi=0\) (the case \(\phi=0\) has already been ruled out).

Next, substituting \(\Theta=0\) and taking the hat and dot derivatives of  \eqref{subbe2} and \eqref{subbe5} we obtain (after simplifications)

\begin{eqnarray}\label{subeas2}
\hat{\dot{\phi}}=\frac{1}{2}\Sigma\left[-2\phi^2-\Sigma^2-\frac{2}{3}\rho-\left(\mathcal{E}+\frac{1}{2}\Pi\right)\right],
\end{eqnarray}
and

\begin{eqnarray}\label{subeas5}
\dot{\hat{\phi}}=\frac{1}{2}\Sigma\left[-\phi^2-2\Sigma^2-\frac{1}{3}\rho-3p+\left(\mathcal{E}+\frac{1}{2}\Pi\right)\right].
\end{eqnarray}
Subtracting \eqref{subeas5} from \eqref{subeas2} we obtain

\begin{eqnarray}\label{subeas8}
\hat{\dot{\phi}}-\dot{\hat{\phi}}=\Sigma\left[\frac{1}{2}\Sigma^2-\frac{1}{2}\phi^2-\frac{1}{6}\rho+\frac{3}{2}p-\left(\mathcal{E}+\frac{1}{2}\Pi\right)\right].
\end{eqnarray}
Now, using the commutation relation \eqref{ghh1} we have

\begin{eqnarray}\label{subeas11}
\hat{\dot{\phi}}-\dot{\hat{\phi}}=-\Sigma\left[\frac{1}{2}\phi^2+\Sigma^2+\frac{2}{3}\rho+\left(\mathcal{E}+\frac{1}{2}\Pi\right)\right].
\end{eqnarray}
Comparing \eqref{subeas8} to \eqref{subeas11} we have the following constraint:

\begin{eqnarray}\label{subeas14}
0=\Sigma\left[\frac{3}{2}\Sigma^2+\frac{1}{2}\left(\rho+3p\right)\right].
\end{eqnarray}
Hence from \eqref{subeas14} we have that either \(\Sigma=0\) or \(\frac{3}{2}\Sigma^2+\frac{1}{2}\left(\rho+3p\right)=0\) in which case either the strong energy condition is violated, i.e. \(\rho+3p<0\) (in this case, if the weak energy condition is to be satisfied, i. e. \(\rho+p>0\), then it is not very difficult to see that \(p<0\), and we are not interested in this case), or that the shear scalar \(\Sigma\in \mathbb{C}\), the set of complex numbers, which we can rule out. If \(\Sigma=0\), then from the vanishing of \(\frac{2}{3}\Theta-\Sigma+\phi\) we also have that \(\phi=0\) on \(\Xi\) (since \(\Theta\) is also zero), and we have already ruled out the case \(\phi=0\). Consequently, we rule out solutions at time \(t= T\vline_{(A=0,\alpha=\beta)}\).\qed
\end{proof}

\subsubsection{Compact dynamical horizons in LRS II spacetimes foliated by marginally trapped surfaces with \(K<0\)}

Next, we consider the case of \(c=2\), where the marginally trapped surfaces are \(2\)-surfaces on which \(K<0\). Comparing \eqref{theq1} for \(c=2\) and \eqref{subbe5} we obtain

\begin{eqnarray}\label{theq18}
\begin{split}
0&=\left(\frac{2}{3}\Theta+\frac{1}{2}\Sigma\right)\left(\frac{2}{3}\Theta-\Sigma\right)-\frac{4}{3}\rho-\frac{1}{2}\phi^2+\mathcal{E}+\frac{1}{2}\Pi,
\end{split}
\end{eqnarray}
which, on the horizon, simplifies \(K\) as 
\begin{eqnarray}\label{theq19}
\begin{split}
K&=-\frac{1}{2}\phi^2-\rho+\left(\frac{2}{3}\Theta-\Sigma\right)\left(\Theta+\frac{1}{2}\phi\right)\\
&=-\phi^2-\rho-\Theta\phi.
\end{split}
\end{eqnarray}
Again, it is sufficient to specify the energy density as positive on \(\Xi\), in which case \(K\) is always negative. Let us now state and prove the following

\begin{theorem}\label{th4} 
Let \(M\) be an LRS II spacetime and let \(\Xi\) be a compact horizon in \(M\) foliated by marginally trapped \(2\)-surfaces with \(K<0\), satisfying \eqref{theq1} for \(c=2\). If the induced metric \(h_{\mu\nu}\) on \(\Xi\) is time dependent, then at some time \(T\in \left(a,t_{end}\right]\) for \(a>t_0\) (where \(t_0\) is the initial time of the metric evolution and \(t_{end}<\infty\)), \(K\) is strictly positive.
\end{theorem}

\begin{proof}
This is easy to show as we note that \eqref{theq10} to \eqref{theq13} holds here as well, and so the implication \(\alpha=\beta\implies K=-\left(1/2\right)\hat{\phi}\) consequently holds. For \(c=2\) in \eqref{theq1} this gives

\begin{eqnarray}\label{theq500}
K=\frac{3}{16}\phi^2,
\end{eqnarray}
which is always positive.\qed
\end{proof}
We see that the situation here gets a little more complicated. The horizon \(\Xi\) is compact if it satisfies \eqref{theq1} for \(c=2\), where \(K<0\) on the marginally trapped surfaces foliating \(\Xi\). By Theorem \ref{th4}, as the metric evolves \(K\) changes sign. For a smooth evolution we have \(K\rightarrow 0\rightarrow \mathbb{R}^{-}\), in which case the Ricci curvature scalar \(R\) on \(\Xi\) is negative. The change of sign of the Gaussian curvature occurs only if the energy density satisfies:

\begin{eqnarray}\label{theq501}
\rho=-\frac{19}{16}\phi^2-\Theta\phi,
\end{eqnarray}
which is always negative. 

A consequence of Theorem \ref{th4} is the following

\begin{corollary}\label{ddfq}
There cannot exist a compact dynamical horizon \(\Xi\) in LRS II spacetimes foliated by marginally trapped \(2\)-surfaces with \(K<0\), satisfying \eqref{theq1} for \(c=2\) where the induced metric \(h_{\mu\nu}\) on \(\Xi\) is time dependent.
\end{corollary}

Corollary \ref{ddfq} is a consequence of a combination of results: As have been discussed in the introduction, Hamilton proved that any smooth closed \(3\)-manifold admitting a metric with \(R>0\) is of spherical geometry. With respect to the induced metric, \(R=-2K>0\) for \(c=2\). Notice that on the marginally trapped surfaces \(^2R=2K<0\) with respect to the induced \(2\)-metric, and hence the marginally trapped surfaces are hyperbolic planes (by Hamilton's proof of the uniformization theorem). Of course hyperbolic planes cannot foliate a \(3\)-sphere. As was mentioned earlier, if \(R> 0\) for the initial metric, then \(R> 0\) throughout the flow. The change of signs of the Gaussian curvature on the marginally trapped \(2\)-surfaces foliating \(\Xi\) from negative to positive at time \(t= T|_{(A=0,\alpha=\beta)}\) implies that \(R\) is now negative (\(R=-2K<0\) for \(K>0\)) at time \(t= T|_{(A=0,\alpha=\beta)}\), which is not possible. 

Note that in Corollary \ref{ddfq} we specified that we are ruling out the existence of compact horizon types in LRS II spacetimes satisfying \eqref{theq1} for \(c=2\) where the induced metric \(h_{\mu\nu}\) on \(\Xi\) is time dependent. While one may speculate that it is therefore possible that there are these horizon types if there is no time dependence of the metric, the first part of the discussion in the previous paragraph has already ruled out this possibility.

The statement of Proposition \ref{th3} combined with that of Corollary \ref{ddfq} may allow us to then state the following: \textit{every compact dynamical horizon of class satisfying \eqref{theq1} in LRS II spacetime, if it exists, is of spherical geometry, i.e.  the geometry \(\mathbb{S}^3\), and their cross sections are topological \(2\)-spheres} (the Gaussian curvature \(K\) remains positive). 

The fact that the Gaussian curvature \(K\) not only determines the geometry of the marginally trapped surfaces, but also characterizes the geometry of the foliated horizons is a very interesting property of the horizon types considered in this work, and it is definitely not trivial.
 
In the next subsection we look to obtain a characterization of compact dynamical horizons satisfying the compactness condition \eqref{theq1} in invariant way.

\subsection{An invariant set characterizing compact dynamical horizons in LRS II spacetimes}

The new covariant way of studying black hole horizons, initiated by \cite{rit1} and extensively exploited by \cite{shef1,shef2}, has proved very useful in unveiling properties of horizons in a relatively straightforward way compared to other approaches. The vanishing of the outgoing null expansion scalar \(\Theta_k\) characterizes marginally trapped surfaces, and by extension the horizons foliated by marginally trapped surfaces. We present an analogue of such characterization of the compact horizon types considered in this work. The primary motive of this construction is to unify, in a formal manner, the relationship between geometry and thermodynamics of the compact dynamical horizon types considered here.

First notice that the condition \eqref{sube416} implies \eqref{sube415} since \(e_{(\mu}u_{\nu)}\neq0\). Interestingly, since \eqref{theq14} implies \(Q=0\), there is a nice geometric condition that gives the condition \(Q=0\): One defines an ``Einstein-like" symmetric \(\left(0,2\right)\)-tensor on \(\Xi\) as \cite{yano1}

\begin{eqnarray}\label{ddr1}
\mathcal{G}_{\mu\nu}=R_{\mu\nu}-\frac{1}{3}Rh_{\mu\nu},
\end{eqnarray}
which has played a crucial role in the study of conformal geometry of Riemannian manifolds \cite{yano1,yano2} (this tensor gives the deviation from Einstein space). In particular, the sign of the norm of this tensor and its contraction with certain complete vector fields have been used to test when a Riemannian manifold is conformorphic or isometric to spheres of varying dimensions. For the cases considered throughout this work, if we calculate the norm of this tensor we obtain

\begin{eqnarray}\label{ddr2}
\begin{split}
\mathbf{G}=\mathcal{G}_{\mu\nu}\mathcal{G}^{\mu\nu}&=R_{\mu\nu}R^{\mu\nu}-\frac{1}{3}R^2\\
&=\frac{2}{3}\left(\alpha-\beta\right)^2.
\end{split}
\end{eqnarray} 
Clearly \(\mathbf{G}=0\) gives \eqref{sube416}, which consequently gives the vanishing of the heat flux, i.e. \(Q=0\). We also have the scalar \(\tilde{\mathbf{G}}=k^{\nu}h^{\mu\delta}\nabla_{\delta}\mathcal{G}_{\mu\nu}\), whose integral plays a crucial role in the study of conformal transformations on Riemannian manifolds of arbitrary dimensions (see the references \cite{yano1,yano2}). Explicitly we calculate

\begin{eqnarray*}
\begin{split}
\tilde{\mathbf{G}}&=\frac{1}{\sqrt{2}}\left(\frac{2}{3}\Theta-\Sigma+\phi\right)K
&=\Theta_kK.
\end{split}
\end{eqnarray*}
Note also that if \(K\neq 0\) (which has been assumed throughout this work), then we must have \(\tilde{\mathbf{G}}=0\implies \Theta_k=0\). Now, define the invariant set \(\mathcal{I}=\lbrace{(\mathbf{G},\tilde{\mathbf{G}})\rbrace}_i\ \text{with}\ (\mathbf{G},\tilde{\mathbf{G}}): \bar{M}\times\bar{M}\to\mathbb{R}^2\), and indexed by the triple \((\Theta,\Sigma,\phi)\) where \(\bar{M}\) is a \(3\)-manifold (Riemannian) in the LRS II class. Then the subset \(\mathcal{J}\subseteq \mathcal{I}\) defined by the constant map \((\mathbf{G},\tilde{\mathbf{G}})\mapsto \left(0,0\right)\), i.e. \(\mathcal{J}=\lbrace{(\mathbf{G},\tilde{\mathbf{G}})\in \mathcal{I}\ \ |\ \  (\mathbf{G},\tilde{\mathbf{G}})\mapsto \left(0,0\right)\in\mathbb{R}^2\rbrace}\) provides a collection of marginally trapped tubes which includes the class of compact dynamical horizons in LRS II spacetimes considered throughout this work. If we assume that the compactness condition \eqref{theq1} is satisfied, then, indeed, the set \(\mathcal{J}\) characterizes the class of compact horizons considered here.

Notice how the scalar \(\tilde{\mathbf{G}}\) not only identifies the horizon for non-zero \(K\), but also that it is sufficient to determine the topology and geometry of the black hole itself: if there is indeed a trapped surface, then \(\Theta_k<0\) by definition. Therefore \(\tilde{\mathbf{G}}<0\) would imply that \(K>0\), in which case the trapped \(2\)-surfaces are spherical, which is as one would expect.

\section{Discussion}\label{soc5}

Initiated purely out of mathematical curiousity, this work set out to investigate the geometry of a certain class of compact dynamical horizons with a time-dependent induced metric in LRS II spacetimes. The geometry of Riemannian manifolds and transformations to metrics on them (conformal, homothetic or isometric) is a well grounded area of study in differential and Riemannian geometry. As mentioned in the introduction, the study of marginally trapped tubes and their evolution, using the \(1+1+2\) semitetrad covariant formalism, has been successfully carried out in recent works \cite{rit1,shef1,shef2}, yielding established results as well as providing clear insights into the nature of the matter and thermodynamic variables on the marginally trapped tubes. Here, we have derived a class of compact horizons and have established geometrical results on these horizons, employing a range of well established results for \(n\)-dimensional compact Riemannian manifolds.

The compactness condition is established - using the Bonnet-Myers theorem - as the requirement that the sheet expansion, \(\phi\), satisfies the spatial first order differential equation \eqref{theq1}, parametrized by a real constant \(c\) which takes on the value `\(0\)' and `\(2\)' in which case the Ricci curvature on the horizon takes the simple form \(R=\pm 2K\), with \(K\) being the Gaussian curvature of \(2\)-surfaces in the spacetime (the `\(+\)' is for `\(c=0\)' and the `\(-\)' for `\(c=2\)'). For the \(c=0\) case it is seen that \(R=2K={}^2R\). 

Let \(\Xi\) be a compact DH of type considered here. Using the Ricci flow evolution equation for compact \(3\)-manifolds it is shown, for the case \(c=0\) (\(K>0\)), that there exists a metric on \(\Xi\) at time \(t= T_{\alpha=\beta}\) for which \(\Xi\) does not radiate, i.e. \(Q=0\). These solutions to the Ricci evolution equations on grounds that, either \(\Xi\) is minimal (this was ruled out since \(\phi=0\) would imply that \(K=0\implies R=0\); this is not possible on a smooth compact Riemannian manifold with positive Ricci curvature) or the shear scalar \(\Sigma\) is complex valued (which can be ruled out), or the strong energy condition (SEC) has to be violated in which case it is  not difficult to show that the weak energy condition is either violated or otherwise the isotropic pressure is negative. Since \(K>0\) at all times of the metric evolution (\(R=2K>0\) at all times), this fact is used to justify the \(\mathbb{S}^3\) geometry of the compact dynamical horizon.

LRS II spacetimes admit no compact dynamical horizons of the type considered here for the case \(c=2\) (\(K<0\)) since if we assume the contrary, then it is shown that \(\Xi\) admits a metric for which \(K>0\) on the marginally trapped surfaces foliating \(\Xi\). This therefore further restricts the number of compact dynamical horizon that may be specified by the compactness condition \eqref{theq1}.

marginally trapped tubes are characterized by the vanishing of the null expansion scalars on surfaces foliating the marginally trapped tubes. Specific conditions are to be satisfied for an marginally trapped tube to be a dynamical horizon, specifically that the marginally trapped tube be spacelike at all points. In some cases additional energy conditions are required to be satisfied. In this work we have provided such a characterization for the compact cases considered. We have used the norm of the Einstein-like tensor \eqref{ddr1} and a certain contraction with the induced metric and the outgoing null normal vector to construct a set consisting of pairs of coordinate independent functions which vanish for the compact dynamical horizons considered here, providing us an invariant way to identify these compact dynamical horizons in LRS II spacetimes.

This work may be seen as initiating an approach to study the geometry and thermodynamics of compact marginally trapped tubes. We intend to, in a follow up work, consider compact marginally trapped tubes (and not just dynamical horizon) in more general spacetime settings, as well as their behavior under conformal rescaling of the metric induced from the ambient spacetime.

\section*{Acknowledgements}

The authors thank the anonymous referee for useful suggestions which helped improve the readibility of the paper. AS acknowledges support from the First Rand Bank, South Africa, through the Department of Mathematics and Applied Mathematics, University of Cape Town. PKSD acknowledges support from the First Rand Bank, South Africa. AS would also like to extend thanks to Dr. Gareth Amery of the School of Mathematics, Statistics and Computer Science, University of KwaZulu-Natal, South Africa, for reading through the first draft of the paper and providing useful comments.

\end{document}